\documentclass[twocolumn]{autart}    

\usepackage{graphicx}          

\usepackage{epsfig} 
\usepackage{times} 
\usepackage{amsmath} 
\usepackage{amssymb}  
\usepackage{tcolorbox}
\usepackage{mathtools}
\usepackage{pgfkeys}
\usepackage{xcolor}
\usepackage{pgfplots}
\usepackage{layouts}
\usepackage{tikz}
\usepackage{algorithm}
\usepackage{mathptmx} 
\DeclareMathAlphabet{\mathcal}{OMS}{cmsy}{m}{n} 
\usepackage{xspace}
\usepackage[T1]{fontenc}
\usepackage{wrapfig}
\usepackage{xurl}
\usetikzlibrary{tikzmark}
\usetikzlibrary{positioning, shapes}
\usetikzlibrary{shapes.symbols}
\usetikzlibrary{calc}

\newlength\figurewidth
\newcommand{\ioss}{i-IOSS\xspace}
\newcommand{\iioss}{i-iIOSS\xspace}

\newcommand{\iiss}{i-iISS\xspace}
\definecolor{new}{rgb}{0,0,0}
\definecolor{rev}{rgb}{0,0,0}
\definecolor{rev2}{rgb}{0,0,0}
\usepackage{nccmath}
\allowdisplaybreaks

\begin{document}

\begin{frontmatter}

\title{	Sample-based detectability and moving horizon state estimation of continuous-time systems} 


\author[IRT]{Isabelle Krauss}\ead{krauss@irt.uni-hannover.de},    
\author[IRT]{Victor G. Lopez}\ead{lopez@irt.uni-hannover.de},               
\author[IRT]{Matthias A. M\"{u}ller}\ead{mueller@irt.uni-hannover.de}  

\address[IRT]{Leibniz University Hannover, Institute of Automatic Control,  30167 Hannover,  Germany}  


	\begin{abstract}
	In this paper we propose a detectability condition for nonlinear continuous-time systems with  irregular/infrequent output measurements,  namely a sample-based version of incremental integral  input/output-to-state stability (\iioss).  We provide a sufficient condition for an \iioss    system to be sample-based \iioss. This condition is also exploited to analyze the relationship between sample-based \iioss and sample-based observability for linear  systems, such that previously established sampling strategies for linear systems can be used to guarantee sample-based \mbox{\iioss}.
	Furthermore, we present a sample-based moving horizon estimation scheme, for which robust stability can be shown. Finally, we illustrate the applicability of the proposed estimation scheme through a  biomedical simulation example.	
\end{abstract}

\end{frontmatter}

	\section{Introduction}
In numerous real-world applications, direct measurement of all system states is impractical or impossible due to physical or technical limitations. However, accurate knowledge of the internal state is essential for tasks such as system monitoring and the design of state feedback controllers. State estimators address this challenge by reconstructing the unmeasured states from the known inputs and measured outputs.
 An optimization-based solution to the state estimation problem is moving horizon estimation (MHE). MHE is a state estimation technique that solves an 
optimization problem based on the past inputs and outputs over a sliding window. Its strength lies in its capability to directly handle nonlinear systems and  constraints, and to cope with model inaccuracies and measurement noise \cite{Raw24}. Strong stability results for MHE have been obtained (compare e.g. \cite{All21,Ji16,Knu23,Sch23,Ale25}) exploiting the nonlinear detectability notion of incremental input/output-to-state-stability (\ioss), first introduced in \cite{Son97}. Recently also a robustly stable MHE scheme for continuous-time (CT) systems has been developed in \cite{Sch24} considering an integral variant of \ioss  \cite{Sch23b}.
\par While standard estimation methods rely on continuous measurements, in many applications we have  only limited output information available.
Such situations may occur, for example, in biomedical applications, where measurements collected through blood samples are available only sparsely or irregularly. This is an issue, for instance, when determining certain hormone concentrations for diagnosing disorders of the pituitary–thyroid  feedback loop \cite{Die16} and devising appropriate medication strategies (see, e.g., \cite{Bru21,Wol22}).
Hence, appropriate state estimation methods are needed that can handle such irregular measurement sequences. 
 To design such estimators, suitable sample-based observability/detectability conditions are needed that take these irregular/infrequent output sequences directly into account. For linear systems, conditions on the sampling sequence to guarantee that the initial state can be uniquely reconstructed based  on the sampled outputs are addressed in \cite{Wan11,Zen16} for CT systems and in \cite{Kra22} for discrete-time (DT) systems.
For nonlinear  DT systems a sample-based formulation of \mbox{\ioss}, that directly takes the irregular output sequence into account, was studied in \cite{Kra25}.
 \par  Exploiting the results in \cite{Kra25}, a DT sample-based MHE scheme for nonlinear systems  was developed in \cite{Kra25d}.
For linear CT systems, a least squares estimation approach for irregular measurement sequences was proposed in \cite{Wan11}.
\color{rev}
Irregular and potentially very infrequent measurement sequences also arise in the event- and self-triggered estimation literature, see, e.g., \cite{Pet24,Niu17} and \cite{And15} for  event-triggered and self-triggered  state estimation methods for CT nonlinear  systems, respectively. In this context, measurement sparsity is not an exogenous constraint, but an online decision dictated by a triggering condition; in event-triggered methods, this involves obtaining the output continuously.
In this work, on the other hand, we exploit a sample-based detectability condition on measurement sequences to construct an estimator for systems in which outputs are only sparsely measurable. 
\color{black}
\par In particular, our contributions are as follows. First, we propose a sample-based version of  incremental integral input/output-to-state stability (\iioss) from \cite{Sch23b} (Section~\ref{sec:sbiIOSS}).  Then, we derive a sufficient condition for an \iioss system to be sample-based \iioss, analogously to the sufficient condition in \cite[Theorems 1,3]{Kra25} for DT systems (Section~\ref{sec:sbiIOSS}). 
Furthermore, we investigate the relationship between sample-based observability and sample-based \mbox{\iioss} for linear systems, where the aforementioned sufficient condition is exploited (Section~\ref{sec:linSys}). Moreover, we propose a robustly stable  sample-based MHE scheme (Section~\ref{sec:sbMHE}). Finally, the presented estimation scheme is applied to the pituitary-thyroid feedback loop (Section~\ref{sec:Sim}). 
\par \textit{Notation:}
Let $\mathbb{R}_{\geq0}$, $\mathbb{R}_{>0}$ be the nonnegative and positive real numbers,  and  $\mathbb{I}_{\geq 0}$, $\mathbb{I}_{>0}$ 
the nonnegative and positive integers, respectively.  $P\succ 0$  ($P\succeq 0$) indicates that a symmetric matrix $P=P^\top$ is positive (semi-)definite.
The euclidean norm of vector $x \in \mathbb{R}^n$   is denoted by $|x|$ and $|x|_P^2=x^\top Px$ with  $P\succ 0$. 
For a matrix $A \in \mathbb{R}^{n\times m}$, $|A|$ refers to the spectral norm of $A$.  
For a measurable, locally essentially bounded function $x:\mathbb{R}_{\geq0} \rightarrow \mathbb{R}^n$, the essential sup-norm on an interval $[a,b]$ is defined as $\|x\|_{a:b}\coloneqq \mathrm{ess} \ \sup_{t\in[a,b]}\left|x(t)\right|$.
By $\sigma_{\mathrm{min}}(P)$  we indicate  the minimum singular value of some matrix $P$. We denote   the  maximum generalized eigenvalue of symmetric 
matrices $P,Q$ by $\lambda_{\mathrm{max}}(P,Q)$.
A function $\alpha:\mathbb{R}_{\geq0} \rightarrow \mathbb{R}_{\geq0}$ is of class $\mathcal{K}_{\infty}$ if it is continuous, strictly increasing, $\alpha(0)=0$, and
$\lim_{s \rightarrow \infty}\alpha(s)=\infty$.
\section{Preliminaries and setup}
\label{sec:Pre}
 We consider the CT nonlinear system 
\begin{subequations}
	\label{eq:sys}
	\begin{align}
			\dot{x}(t)&=f(x(t),u(t),w(t)) \label{eq:sysf}\\
			y(t)&=h(x(t),u(t),v(t)) \label{eq:sysh}
	\end{align}
\end{subequations}
with state $x(t) \in \mathcal{X} \subseteq \mathbb{R}^n$, noisy output measurement $y(t) \in \mathcal{Y} \subseteq \mathbb{R}^p$ and time $t\geq 0$. The  control input $u$, process noise $w$ and measurement noise $v$ are measurable, locally essentially bounded functions taking values in $\mathcal{U} \subseteq \mathbb{R}^m$, $\mathcal{W} \subseteq \mathbb{R}^{q_w}$ (with $0 \in \mathcal{W}$) and $\mathcal{V} \subseteq \mathbb{R}^{q_v}$ (with $0 \in \mathcal{V}$), respectively. The sets of such functions $u$, $w$ and $v$ and are denoted by $\mathcal{M_U}$, $\mathcal{M_W}$ and $\mathcal{M_V}$, respectively. The  nonlinear continuous functions $f: \mathcal{X} \times \mathcal{U} \times \mathcal{W} \rightarrow \mathbb{R}^n$ and $h: \mathcal{X} \times \mathcal{U} \times \mathcal{V}\rightarrow \mathcal{Y}$ represent the system dynamics and the output model, respectively. Given an initial state $\chi\in\mathcal{X}$, control input $u\in \mathcal{M_U}$ and disturbance input $w\in \mathcal{M_W}$,  we denote  the solution to (\ref{eq:sysf}) at time $t\geq 0$ as $x(t,\chi,u,w)$ and we use the short-hand notation $x(t)$ when $\chi, u, w$ are clear from the context. \color{rev} We assume $x(t)\in \mathcal{X}$ for all $t\geq0$, all $u\in \mathcal{M_U}$ and all  $w\in \mathcal{M_W}$\footnote{This reflects known physical constraints  such as nonnegativity of state variables like concentrations or pressures. Such knowledge is exploited in MHE to improve the estimation result.}. \color{black} 
Moreover, the output signal is denoted as $y(t,\chi,u,w,v)\coloneq h(x(t,\chi,u,w),u(t),v(t))$. We assume that $\mathcal{X}$, $\mathcal{W}$ and $\mathcal{V}$ are closed and that  the solutions of (\ref{eq:sysf}) are unique and globally defined on $\mathbb{R}_{\geq 0}$ for all $u \in  \mathcal{M_U}$, $w  \in \mathcal{M_W}$  and $\chi \in\mathcal{X}$.
In the following we use the  abbreviated notation 
\begin{align*}
			\begin{aligned}
					\Delta w(t) &\coloneq w_1(t) - w_2(t),\  \Delta v(t) \coloneq v_1(t) - v_2(t)\\
					\Delta x(t) &\coloneq x(t,\chi_1, u, w_1) - {x}(t, \chi_2, u, w_2 ), \ 	 \Delta\chi \coloneq  \chi_1-\chi_2,
					\\	\Delta y(t) &\coloneq y(t,\chi_1,u,w_1,v_1) - y(t, \chi_2, u,w_2,v_2). 
			\end{aligned}
\end{align*}
   Incremental input/output-to-state-stability (\ioss) has become a standard detectability notion in the context of optimization-based state estimation, compare, for example \cite{Raw12,All21,Raw24,Sch23,Hu24}, and was shown to be a necessary (cf. \cite{Son97}) and sufficient (cf. \cite{Knu23}) condition for the existence of robustly stable estimators. In this work we consider an integral variant of \ioss (see Definition~\ref{def:iiioss} below), namely incremental integral input/output-to-state stability (\iioss),  as was proposed in \cite{Sch23b}. Note that this definition of \iioss discounts  the influence of the inputs and outputs in the past. Using such time-discounting allowed the establishment of strong stability results for MHE which, in the DT case, led to tighter upper bounds on the estimation error compared to earlier approaches based on non-time-discounted \ioss \cite{Sch23}. In \cite{Sch23b} it was shown that \iioss is equivalent to the system admitting an \iioss Lyapunov function. A method to construct such an \iioss Lyapunov function has been proposed in \cite{Sch24}, and thus providing a systematic approach to verify \iioss. Furthermore, \iioss is necessary for the existence of  a robustly globally asymptotically stable
observer mapping in a time-discounted "$L^2$-to-$L^{\infty}$" sense \cite{Sch23b}.  
\begin{defn}[\protect{\iioss}]
	\label{def:iiioss}
	The system~\eqref{eq:sys} is   incrementally integral input/output-to-state stable (\iioss) if there exist some $\alpha,\alpha_{x},\gamma_1,\gamma_2,\gamma_3\in\mathcal{K}_{\infty}$, and $\eta\in \mathbb{R}_{> 0}$ such that for any initial conditions  $\chi_1,\chi_2$ and any pairs of disturbance trajectories $w_1,w_2\in\mathcal{M_W}$, $v_1,v_2\in\mathcal{M_V}$ and all $u\in \mathcal{M_U}$ 
	the following holds for all $t \geq 0$
	\begin{align}
		\hspace{-0.8em}
			\begin{aligned}
				\alpha({|}\Delta x(t){|})
				&\leq \alpha_x({|}\Delta \chi{|})e^{-\eta t}+\int_0^t\Big(\gamma_1({|}\Delta w(\tau){|})  \\&+\gamma_2({|}\Delta y(\tau){|})+\gamma_3({|}\Delta v(\tau){|}) \Big) e^{-\eta(t-\tau)}  d\tau.
			\end{aligned}
		\label{eq:ioss}
	\end{align}
	\color{rev}
	If additionally 	$\alpha(s) \geq  C_1 s^r$, $\alpha_x(s) \leq C_2 s^r$ , $\gamma_1(s) \leq C_w s^r$, $\gamma_2(s) \leq C_y s^r$ and $\gamma_3(s) \leq C_v s^r$ for some $C_1, C_2,C_w,C_y,C_v, r >0$, then system~\eqref{eq:sys} is exponentially \iioss.
\end{defn}
\color{black} Moreover, if (\ref{eq:ioss}) holds for some $\alpha,\alpha_x,\gamma_1\in \mathcal{K}_\infty$ and $\gamma_2=\gamma_3 \equiv 0$, the system is said to be  incrementally integral input-to-state stable (\iiss) \cite{Ang09}.
\section{Sample-based \iioss}
\label{sec:sbiIOSS}
In the following we introduce a sample-based version of  \mbox{\iioss}.  Then we propose a sufficient condition for an \mbox{\iioss} system to be sample-based \iioss, analogously to the  sufficient condition in \cite[Theorems 1,3]{Kra25} for the DT case.
Before formulating a sample-based version of \iioss, we recall the following definition from \cite{Kra25}.
\begin{defn}[Sampling set $K$ \cite {Kra25}]
	\label{def:K}
	Consider an infinitely long sequence $\mathcal{D}=\{d_1,d_2,\ldots\}$ with $d_i \geq 0, \ i\in \mathbb{I}_{> 0}$ and  $\max_i d_i \eqcolon d_{\mathrm{max}}<\infty$. 
	The set $K_i=\{t_1^i,t_2^i,\ldots\}$ refers to an infinite set of time instances defined as
	\begin{align*} t_1^i = d_i, \quad t_2^i = t_1^i + d_{i+1}, \ \dots \, \quad t_j^i = t_{j-1}^i + d_{i+j-1}, \  \dots \end{align*}
	The set $K$ then refers to a set of sets containing all $K_i$, $ i\in\mathbb{I}_{> 0}$.
\end{defn}
\color{rev} The sampling set $K$ is defined as a set of sets, rather than a single sequence of sampling times, to capture sampling schemes that preserve detectability uniformly in time. 
In this sense, we are interested in sampling schemes for which sample-based detectability holds for all $K_i \in K$. 
This means that the state of the system can be reconstructed regardless of the time at which we start using the measured information. For a more detailed discussion see \cite{Kra25}.
\color{rev2} An exemplary illustration of how the time sequences corresponding to $K_1$ and $K_2$ are generated from $\mathcal{D}$ is  provided below in Figure~\ref{fig:K}. \color{black}  We now define sample-based \iioss as follows.
\begin{figure}[h]
	\centering
	\scalebox{0.9}{
	\begin{tikzpicture}
		\centering
	
		\draw[->] (0,0) -- (8,0);
		
		\draw[black] (0,0.15) -- (0,-0.15);
		\draw[blue!75!black] (1.2,0.15) -- (1.2,-0.15);
		\draw[blue!75!black] (4,0.15) -- (4,-0.15);
		\draw[blue!75!black] (5.6,0.15) -- (5.6,-0.15);

		\node[blue!75!black, below] at (1.2,-0.15) {$t_1^1$};
		\node[blue!75!black, below] at (4,-0.15) {$t_2^1$};
		\node[blue!75!black, below] at (5.6,-0.15) {$t_3^1$};

		\draw[blue!75!black, decorate, decoration={brace,amplitude=5pt}]
		(0,0.35) -- (1.2,0.35)
		node[midway,above=5pt] {$d_1$};
		
		\draw[blue!75!black, decorate, decoration={brace,amplitude=5pt}]
		(1.2,0.35) -- (4,0.35)
		node[midway,above=5pt] {$d_2$};
		
		\draw[blue!75!black, decorate, decoration={brace,amplitude=5pt}]
		(4,0.35) -- (5.6,0.35)
		node[midway,above=5pt] {$d_3$};

		\draw[green!60!black,thick] (2.8,0.15) -- (2.8,-0.15);
		\draw[green!60!black] (4.4,0.15) -- (4.4,-0.15);
		\draw[green!60!black] (6.4,0.15) -- (6.4,-0.15);
		
		\node[green!60!black, below] at (2.8,-0.15) {$t_1^2$};
		\node[green!60!black, below] at (4.4,-0.15) {$t_2^2$};
		\node[green!60!black, below] at (6.4,-0.15) {$t_3^2$};
		
		\draw[green!60!black, decorate, decoration={brace,amplitude=5pt}]
		(0,1.1) -- (2.8,1.1)
		node[midway,above=5pt] {$d_2$};
		
		\draw[green!60!black, decorate, decoration={brace,amplitude=5pt}]
		(2.8,1.1) -- (4.4,1.1)
		node[midway,above=5pt] {$d_3$};
		
		\draw[green!60!black, decorate, decoration={brace,amplitude=5pt}]
		(4.4,1.1) -- (6.4,1.1)
		node[midway,above=5pt] {$d_4$};

	\end{tikzpicture}
}
	\caption{Illustration of  $\textcolor{blue!75!black}{K_1=\{t_1^1,t_2^1,t_3^1,\ldots\}}$ and $\textcolor{green!60!black}{K_2=\{t_1^2,t_2^2,t_3^2,\ldots\}}$.}
	\label{fig:K}
\end{figure}
\vspace{-0.4em}
\begin{defn}[Sample-based \iioss]
	\label{def:sbIOSS}
	Consider a set $K$ as in Definition~\ref{def:K}.
	The system~(\ref{eq:sys}) is sample-based  \iioss with respect to $K$ if there exist  $\bar{\alpha},\bar\alpha_x,\bar\gamma_1, \bar\gamma_2,\bar\gamma_3 \in  \mathcal{K}_\infty$ and $\bar\eta\in \mathbb{R}_{> 0}$ such that for any pair of initial conditions  $\chi_1,\chi_2$, any pairs of disturbance trajectories  $w_1,w_2 \in \mathcal{M_W}$, $v_1,v_2 \in \mathcal{M_V}$ and all $u\in\mathcal{M_U}$ the following holds for all $t\geq 0$ and any $K_i\in K$
	\begin{align}
			&\bar\alpha({|}\Delta x(t){|})
			\leq \bar\alpha_x({|}\Delta \chi{|})e^{-\bar\eta t}+\int_0^t\bar\gamma_1({|}\Delta w(\tau){|}) e^{-\bar\eta(t-\tau)} \ d\tau \nonumber
			\\
			&+\sum_{\tau \in [0,t] \cap K_i}\big(\bar\gamma_2( {|}\Delta y(\tau){|})+\bar\gamma_3({|}\Delta v(\tau){|})\big) e^{-\bar\eta(t-\tau)}.\label{eq:sbiIOSS}
	\end{align}
	\color{rev}
		If additionally 	$\bar\alpha(s) \geq  \bar{C}_1 s^r$, $\bar\alpha_x(s) \leq \bar{C}_2 s^r$,  $\bar\gamma_1(s) \leq \bar C_w s^r$, $\bar \gamma_2(s) \leq \bar C_y s^r$ and $\gamma_3(s) \leq\bar C_v s^r$ for some $\bar C_1, \bar C_2, \bar C_w,\bar C_y, \bar C_v, r >0$, then  system~(\ref{eq:sys}) is sample-based exponentially \iioss.
		\color{black}
\end{defn} 
Since this sample-based detectability notion only considers discrete measurements, the output-related terms in (\ref{eq:sbiIOSS}) use a sum instead of an integral.
\color{rev} %
This definition states that, if a system is sample-based \iioss, 
then the differences between two system states can be bounded by the differences in their initial states, inputs and \emph{sampled} outputs. 
 Intuitively, if two trajectories experience similar inputs but  their internal states diverge significantly, then this divergence must be captured in their measured outputs. 
\color{black}
\par  For deriving a sufficient condition for sample-based \iioss,  in Theorem~\ref{thm:sufcond} below we  use the following assumption on $f$, that is  Lipschitz continuity of $f$ w.r.t $x$ and $\mathcal{K}$-continuity w.r.t. to $w$.
\begin{assum}
\label{ass:f}
There exist some $\rho \in \mathcal{K_\infty}$ and some $L>0$ such that for all $u\in \mathcal{U}$, all $x_1, x_2 \in\mathcal{X}$ and all $w_1,w_2 \in \mathcal{W}$ it holds that
\begin{align*}
	|f (x_1, u,w_1) - f (x_2,u, w_2)| \leq L(|x_1 - x_2| + \rho(|w_1 - w_2|)).
\end{align*}
\end{assum}
After having stated Assumption~\ref{ass:f} we can now derive the following theorem.
Note that an analogous sufficient condition for sample-based \ioss has been proposed in \cite[Theorems~1,3]{Kra25} for the DT case. 
\color{rev2}
While the main intuition behind this sufficient condition remains the same,
 Theorem~\ref{thm:sufcond} below requires different arguments in the proof to account for the CT nature of the setting in this work. In particular, the main novelties in this proof are (i) exploiting Assumption~\ref{ass:f} to bound state trajectories, (ii) partitioning the time horizon into several intervals and bounding the output trajectories over each interval, and (iii) leveraging different properties of $\mathcal{K}_\infty$ functions, including their decomposition into concave and convex functions, to handle the particularities of the integrals arising in the CT case. \color{black} 
Intuitively, condition (\ref{eq:suffcond}) in Theorem~\ref{thm:sufcond} means that there exists some finite time at which a sufficient number of meaningful measurements have been collected, after which any changes  in the system's output can be detected from the limited number of available measurements.
For a more detailed discussion see the remarks below \cite[Theorem~1]{Kra25}. 
\begin{thm}  [Sufficient Condition]
\label{thm:sufcond}
Let system (\ref{eq:sys})  be \iioss and let Assumption \ref{ass:f} hold. Consider some set $K$ according to Definition~\ref{def:K}. The system is sample-based  \iioss if 
there exist a finite $t^*$ and $\alpha_w, \alpha_y, \alpha_v \in \mathcal{K}_{\infty}$ such that for any two initial conditions $\chi_1, \chi_2$ and any pairs of disturbance trajectories $w_1,w_2\in \mathcal{M_W}$, $v_1,v_2\in \mathcal{M_V}$ and all $u\in \mathcal{M_U}$ the following holds for all $t\geq  t^*$ and all $K_i\in K$
\begin{align}
		\label{eq:suffcond}
	\begin{aligned}
		\int_{t^*}^t\gamma_2( {|}\Delta y(\tau){|})+\gamma_3( {|}\Delta v(\tau){|})  \ d\tau
		\leq \int_0^t\alpha_w({|} \Delta w(\tau)
		{|}) \ d\tau\\+ \sum_{\tau \in [0,t] \cap K_i}\alpha_y( {|}\Delta y(\tau){|})+\alpha_v( {|}\Delta v(\tau){|}).
	\end{aligned}
\end{align}
\end{thm}
\begin{pf}
Due to Assumption~\ref{ass:f} we have for all $t\in [0,t^*]$
\begin{align}
	{|}\Delta x(t){|}\leq  \Big(|\Delta \chi{|}+\int_0^{t}L\rho({|}\Delta w(\tau){|}) \ d\tau\Big)e^{Lt}.	
	\label{eq:t<tstar}	
\end{align}
This can be readily seen by suitably adapting  Theorem~1.41 and its proof in \cite{Mir23}. 
Hence, 
\begin{align}
		\begin{aligned}
				\alpha_x(	{|}\Delta x(t^*){|})\leq &	\alpha_x(2|\Delta \chi{|}e^{Lt^*})\\&+	\alpha_x\Big(2\int_0^{t^*}L\rho({|}\Delta w(\tau){|}) \ d\tau  \  e^{Lt^*}\Big).	
		\end{aligned}
	\label{eq:betarho}	
\end{align}
By the \iioss bound (\ref{eq:ioss}) we know that for all $t\geq t^*$
\begin{align}
	\begin{aligned}
		&	\alpha({|}\Delta x(t){|})\leq  \alpha_x({|} \Delta x(t^*)
		{|}) e^{-\eta(t-t^*)}	\label{eq:iiosststar}
		+\int_{t^*}^t\Big(\gamma_1({|}\Delta w(\tau){|})  \\&+\gamma_2({|}\Delta y(\tau){|})+\gamma_3({|}\Delta v(\tau){|}) \Big) e^{-\eta(t-\tau)} \ d\tau.
	\end{aligned}
\end{align}
Now fix some arbitrary $K_i \in K$ and split the interval $[t^*,t],t>t^*$ in $\vartheta\in \mathbb{I}_{>0}$ time intervals  $[k_j,k_{j+1}]$ such that (i) $t^*\leq k_{j+1}-k_j\leq 2t^*$, $j=1,\ldots, \vartheta$ with $k_1=t^*$ and $k_{\vartheta+1}=t$ and (ii) $k_j-{t^*} \in K_i$ for  $j=2,\ldots, \vartheta$. We can always find such $k_j$ for each $K_i\in K$ 
by selecting $t^*$ large enough (which can be done without loss of generality) and due to the difference between any two measurements being bounded by $d_{\mathrm{max}}$ (cf. Definition~\ref{def:K}).
Using (\ref{eq:suffcond}) we can derive the following bound for each of the $\vartheta$ time intervals
\begin{align}
		&e^{-\eta(t-k_{j+1})}\int_{k_j}^{k_{j+1}} \gamma_2( {|}\Delta y(\tau){|})+\gamma_3( {|}\Delta v(\tau){|})  \ d\tau \nonumber\\ &\leq e^{-\eta(t-k_{j+1})}\Big(\int_{k_j-t^*}^{k_{j+1}}\alpha_w({|} \Delta w(\tau)
		{|})e^{-\eta(k_{j+1}-\tau)}  e^{\eta 3t^*} \ d\tau  \nonumber \\&+\hspace{-0.3em}\sum_{\substack{\tau \in [k_j-t^*,k_{j+1}] \\ \cap K_i}} \hspace{-0.6em}\big(\alpha_y( {|}\Delta y(\tau){|})+\alpha_v( {|}\Delta v(\tau)|)\big)e^{-\eta(k_{j+1}-\tau)}  e^{\eta 3t^*}\Big) \nonumber\\
		&	\leq  e^{\eta 3t^*} \int_{k_j-t^*}^{k_{j+1}}\alpha_w({|} \Delta w(\tau)
		{|})e^{-\eta(t-\tau)} \ d\tau \label{eq:bound_nu_interval}\\
		&+  e^{\eta 3t^*}
		\sum_{\substack{\tau \in [k_j-t^*,k_{j+1}] \\ \cap K_i}} \hspace{-0.6em}\big(\alpha_y( {|}\Delta y(\tau){|})+\alpha_v( {|}\Delta v(\tau)|)\big)e^{-\eta(t-\tau)}.  \nonumber
\end{align}
Therefore, and due to 
\begin{align} 
		\begin{aligned}
			&\int_{k_j}^{k_{j+1}} \big(\gamma_2( {|}\Delta y(\tau){|})+\gamma_3( {|}\Delta v(\tau){|}) \big)e^{-\eta(t-\tau)} \ d\tau \\\leq& e^{-\eta(t-k_{j+1})}\int_{k_j}^{k_{j+1}}  \big(\gamma_2( {|}\Delta y(\tau){|})+\gamma_3( {|}\Delta v(\tau){|}) \big) \ d\tau
		\end{aligned}
	\label{eq:interval_each}
\end{align}
we obtain for all $t\geq t^*$
\begin{align}
	&	\int_{t^*}^{t} \big(\gamma_2( {|}\Delta y(\tau){|})+\gamma_3( {|}\Delta v(\tau){|}) \big)e^{-\eta(t-\tau)}  d\tau 	\nonumber	
	\\& \leq 2e^{\eta 3t^*}\int_{0}^{t}\alpha_w({|} \Delta w(\tau)
	|)e^{-\eta(t-\tau)} \ d\tau  \label{eq:tdcond}
	\\&+ 2e^{\eta 3t^*} \sum_{\tau \in [0,t] \cap K_i}\big(\alpha_y( {|}\Delta y(\tau){|})+\alpha_v( {|}\Delta v(\tau)|)\big)e^{-\eta(t-\tau)}. \nonumber
\end{align} 
Note that when using (\ref{eq:bound_nu_interval}) and (\ref{eq:interval_each}) to obtain an upper bound for the whole time interval $[t^*,t]$ in (\ref{eq:tdcond}), the integration bounds on the right-hand side of (\ref{eq:bound_nu_interval}) overlap at most by $t^*$ (since $t^*\leq k_{j+1}-k_j$). This results in the factor $2$ in (\ref{eq:tdcond}).
Substituting (\ref{eq:betarho}) and  (\ref{eq:tdcond}) in (\ref{eq:iiosststar}) yields for all $t\geq t^*$
\begin{align}
\hspace{-1em}
	&\alpha({|}\Delta x(t){|})\leq \gamma_x({|}\Delta \chi{|})e^{-\eta t}+\int_{t^*}^t\gamma_1({|} \Delta w(\tau)
	{|})e^{-\eta(t-\tau)} \ d\tau \nonumber
	\\&+\alpha_x\Big(\int_0^{t^*} 2Le^{Lt^*}\rho({|}\Delta w(\tau){|})  \ d\tau\Big)e^{-\eta(t-t^*)}  \nonumber
	\\&+\int_0^t 2e^{\eta 3 t^*}\alpha_w({|} \Delta w(\tau)
	{|})e^{-\eta(t-\tau)} \ d\tau 	\label{eq:combbound}
	\\&+ \sum_{\tau \in [0,t] \cap K_i}2e^{\eta 3 t^*}\big(\alpha_y( {|}\Delta y(\tau){|})+\alpha_v( {|}\Delta v(\tau)|)\big)e^{-\eta(t-\tau)} \nonumber
\end{align}
with $\gamma_x(s)\coloneq\alpha_x(2e^{Lt^*} s)e^{\eta t^*}$. 
Note that we can upper bound $\alpha_x(s)\leq\gamma_c(\gamma_v(s))$ with  $\gamma_c$ and $\gamma_v$ being a concave and convex $\mathcal{K}_{\infty}$-function, respectively (cf. \cite[Lemma~13]{Kel14}). Applying Jensen’s inequality to  $\gamma_v$ 
and exploiting concavity of  $\gamma_c$	together with $e^{-\eta(t-t^*)}\in[0,1]$ for $t\geq t^*$
we can upper bound the third term on the right-hand side of (\ref{eq:combbound}) by
\begin{align}
\begin{aligned}
	&	\alpha_{x}\Big(\int_0^{t^*} 2Le^{Lt^*}\rho({|}\Delta w(\tau){|}) \ d\tau\Big)e^{-\eta(t-t^*)} \\ \leq &\gamma_c\Big(\frac{1}{t^*}\int_0^{t^*}\gamma_v(t^*2Le^{Lt^*} \rho({|}\Delta w(\tau){|}))e^{\eta\tau} d \tau\Big)e^{-\eta(t-t^*)} \\\leq& \gamma_c\Big(\frac{1}{t^*}\int_0^{t^*}\gamma_v( t^*2Le^{Lt^*}\rho({|}\Delta w(\tau){|}))e^{\eta  t^*}e^{-\eta(t-\tau)} d \tau\Big)
\end{aligned}
\label{eq:combbound_3term}
\end{align}
Note that the term $e^{\eta \tau}$ could be added in the second line since it is always greater than or equal to one.
Furthermore, note that $\alpha_x(s)\geq \alpha(s), \  \forall s\in \mathbb{R}_{\geq 0}$ (which follows from (\ref{eq:ioss}) with $t=0$) and that $e^{-\eta(t-t^*)}\geq 1$ for $t\in[0,t^*]$. Thus, using (\ref{eq:t<tstar}) we can write 
\begin{align}
	\label{eq:alphat<tstar}	
\begin{aligned}
	&	\alpha(	{|}\Delta x(t){|})\leq 	\alpha_x(2|\Delta \chi{|}e^{Lt^*})e^{-\eta(t-t^*)}\\&+	\alpha_x\Big(2\int_0^{t}L\rho({|}\Delta w(\tau){|}) \ d\tau  \  e^{Lt^*}\Big).	
\end{aligned}
\end{align}
Hence, it can be seen that substituting (\ref{eq:combbound_3term}) in (\ref{eq:combbound}) not only provides an upper bound for $t\geq t^*$ but also upper bounds  (\ref{eq:alphat<tstar})	for $t\leq t^*$. Therefore, 
we have  an upper bound that holds for all $t\geq 0$. 
This we can further bound as follows 
\begin{align*}
\begin{aligned}
	&\alpha({|}\Delta x(t){|})\leq \hat\alpha\Big(\gamma_x({|}\Delta \chi{|})e^{-\eta t}
	+\int_{t^*}^t\gamma_1(| \Delta w(\tau)
	|)e^{-\eta(t-\tau)} \ d\tau\\&+\frac{1}{t^*}\int_0^{t^*}\gamma_v(t^* 2Le^{Lt^*}\rho({|}\Delta w(\tau){|}))e^{\eta  t^*} e^{-\eta(t-\tau)} \ d\tau\\&+\int_0^t 2e^{\eta 3 t^*}\alpha_w({|} \Delta w(\tau)
	{|})e^{-\eta(t-\tau)} \ d\tau\\&+ \sum_{\tau \in [0,t] \cap K_i}2e^{\eta 3 t^*}\big(\alpha_y( {|}\Delta y(\tau){|})+\alpha_v( {|}\Delta v(\tau)|)\big)e^{-\eta(t-\tau)}\Big)
\end{aligned}
\end{align*}
with $\hat\alpha(s)=s+\gamma_c(s)$. Then we finally obtain
\begin{align*}
\begin{aligned}
&	\bar{\alpha}({|}\Delta x(t){|})\leq  \bar{\alpha}_x({|}\Delta \chi{|})e^{-\bar\eta t}+\int_0^{t}\bar{\gamma}_1({|}\Delta w(\tau){|}) e^{-\bar\eta(t-\tau)} \ d\tau\\&+ \sum_{\tau \in [0,t] \cap K_i}\big(\bar{\gamma}_2( |\Delta y(\tau)|)+\bar{\gamma}_3( |\Delta v(\tau)|)\big)e^{-\bar\eta(t-\tau)}, \quad \forall t\geq 0
\end{aligned}
\end{align*}
with $\bar{\alpha}(s)=\hat{\alpha}^{-1}(\alpha(s))$ , $\bar{\alpha}_x(s)=\gamma_x(s)$, $\bar{\gamma}_1(s)=\gamma_1(r)+2e^{\eta 3 t^*}\alpha_w(s)+\gamma_v(t^*2Le^{Lt^*}\rho(s))\frac{e^{\eta  t^*}}{t^*}$, $\bar\gamma_2(s)=2e^{\eta 3 t^*}\alpha_y(s)$, $\bar\gamma_3(s)=2e^{\eta 3 t^*}\alpha_v(s)$ and $\bar{\eta}=\eta$, i.e., system (\ref{eq:sys}) is sample-based
\iioss according to Definition~\ref{def:sbIOSS}. 
\hfill $\qed$	
\end{pf}
\color{rev2} From Theorem~\ref{thm:sufcond} it becomes clear that a system that is detectable in the standard sense will also be sample-based detectable if enough meaningful output samples are collected, as stated by the right-hand side of (\ref{eq:suffcond}). For nonlinear systems, verifying condition (\ref{eq:suffcond}) along all system trajectories is generally difficult. Numerical approaches can be followed instead, verifying (\ref{eq:suffcond}) along a set of trajectories of interest to gain empirical information about the quality of a sampling scheme.  In Section~\ref{sec:linSys}, existing results on sample-based observability are used to establish when condition (\ref{eq:suffcond}) holds for linear systems. \color{black}
\color{black}
\section{Sample-based MHE}
\label{sec:sbMHE}
MHE  is an optimization-based state estimation technique. At each time step, an optimization problem is solved over a sliding window, with the past inputs and outputs being taken into account, to obtain the current state estimate. 
In the following, we propose a robustly stable sample-based MHE for CT systems that considers an irregular measurement sequence. 
\subsection{Setup}
\label{sec:MHE_setting}
We define $K_s$ as the sequence  containing the time instances where a measurement is available to the estimator.
In the proposed scheme the optimization problem is not explicitly solved at regular (equidistant) sampling instances, but only when new information is available. 
This means that, for $t_i \in K_s$, an optimization problem is solved at time $t_i$ with the horizon length $M_{t_{i}}\coloneq\min\{t_i,M\}$, \color{rev2}$M >0$. \color{black} The past input and output trajectories considered in the moving window $[t_i-M_{t_i},t_i]$ are denoted by $u_{t_i}$ and $y_{t_i}$, i.e., 
\begin{align*}
	\begin{aligned}
		u_{t_i}(\tau)\coloneq u(t_i-M_{t_i}+\tau), \ \tau \in [0,M_{t_i}]\\
		y_{t_i}(\tau)\coloneq y(t_i-M_{t_i}+\tau, \chi, u,w,v), \ \tau \in [0,M_{t_i}].
	\end{aligned}
\end{align*}
\color{black}
To obtain the current state estimate, the following NLP is solved for $t_i \in K_s$
\begin{align}
\min_{\hat{\chi}_{t_i}, \hat{w}_{t_i},\hat{v}_{t_i}}
&	J(\hat{\chi}_{t_i}, \hat{w}_{t_i}, \hat{v}_{t_i},  \hat{y}_{t_i},t_i) \nonumber \\
\mathrm{s.t.} \  \hat{x}_{t_i}(\tau)&=x(\tau,\hat{\chi}_{t_i},u_{t_i},\hat{w}_{t_i}), \ \tau\in \label{eq:NLP} [0,M_{t_i}], 
\\	\hat{y}_{t_i}(\tau)&=y(\tau,\hat{\chi}_{t_i},u_{t_i},\hat{w}_{t_i},\hat{v}_{t_i}), \ \hat{y}_{t_i}(\tau)  \in \mathcal{Y}, \ \tau\in [0,M_{t_i}],\nonumber \\ 
\hat{w}_{t_i}(\tau) &\in \mathcal{W},\ \hat{v}_{t_i}(\tau) \in \mathcal{V}, \ \hat{x}_{t_i}(\tau)  \in \mathcal{X}, \ \tau\in [0,M_{t_i}] \nonumber 
\end{align}
where $ \hat{w}_{t_i},\hat{v}_{t_i},\hat{y}_{t_i}$  refer to the segments of the  estimated disturbance and output trajectories, respectively, on the time interval $[t_i-M_{t_i},t_i]$. Furthermore, $\hat{\chi}_{t_i}$ denotes the initial estimate for the interval $[t_i-M_{t_i},t_i]$, estimated at time $t_i$.  We denote the optimal solution minimizing  (\ref{eq:NLP}) by $(\chi_{t_i}^*, \hat{w}_{t_i}^*,\hat{v}_{t_i}^*)$ and the corresponding optimal state trajectory by $\hat{x}_{t_i}^*(\tau)=x(\tau,\hat{\chi}_{t_i}^*,u_{t_i},\hat{w}_{t_i}^*)$.  We define $\hat{x}(t_i)\coloneq \hat{x}_{t_i}^*(M_{t_i})$ and the state estimation error is given by $e(t_i)\coloneq x(t)-\hat{x}(t)$. We consider the following cost function 
\begin{align}
&J(\hat{\chi}_{t_i}, \hat{w}_{t_i},   \hat{y}_{t_i},t_i)=2e^{-\eta M_{t_i}}{|}\hat{\chi}_{t_i}-\hat{x}(t_i-M_{t_i}){|}^2_{P_2} \nonumber\\&+\int_{0}^{M_{t_i}} e^{-\eta(M_{t_i}-\tau)}2 {|}\hat{w}_{t_i}(\tau){|}_{Q_w}^2 \ d\tau \nonumber
\\&+\sum_{\tau+t-M_{t_i}\in [t_i-M_{t_i},t_i] \cap K_s} e^{-\eta(M_{t_i}-\tau)}2{|}\hat{v}_{t_i}(\tau){|}_{Q_v}^2	\label{eq:objfunc}
\\&	+\sum_{\tau+t-M_{t_i}\in [t_i-M_{t_i},t_i] \cap K_s} e^{-\eta(M_{t_i}-\tau)} {|}\hat{y}_{t_i}(\tau)-y_{t_i}(\tau){|}_R^2 \nonumber
\end{align}
with $\eta \in\mathbb{R}_{>0}$. 
The first part of the cost function penalizes the difference between the first element of the estimated state sequence ($\hat{\chi}_{t_i}$) and the prior estimate using the weighting matrix $P_2\succ 0$. The remaining three terms consider the estimated disturbance sequences and the difference between the measured and estimated outputs at the time instances where a measurement is available to the state estimator, weighted by $Q_w\succ 0$, $Q_v\succ 0$ and $R\succ 0$, respectively. \color{rev2} Due to the discount factor $e^{-\eta}$, the influence of disturbances
and output measurements further in the past is reduced. \color{black} 
\color{rev}
Further discussion on the parameterization of the cost function can be found in Remark~\ref{rem:para} below.
\color{black}
\par At $t_i \in K_s$, we obtain $\hat{x}(t_i)$ by solving the MHE's optimization problem. The next measurement is available at $t_{i+1}$ and thus $\hat{x}(t_{i+1})$ is computed again by solving (\ref{eq:NLP}). For $t_i<t<t_{i+1}$ the estimator runs in open-loop, i.e., 
\begin{align*}
\hat{x}(\tau)&=x(\tau,\hat{x}(t_i),u,0), \ \tau\in (t_i,t_{i+1}).
\end{align*}
\color{rev2} Note that the computational characteristics of the proposed optimization problem, including complexity, real-time feasibility, and scalability as the horizon increases, are comparable to that of standard MHE. 
However, the proposed approach requires solving the optimization problem less often.\color{black}
\subsection{Robust stability analysis}
In the following we show robust stability of the proposed sample-based MHE scheme according to the following definition.
This definition is adapted to our CT setting with discrete measurements and is thus formulated using both and integral and discrete summation terms.
\begin{defn}
\label{def:rges}
Consider a sequence of measurement instances $K_s$. A state estimator for system (\ref{eq:sys}) starting at initial condition $\hat\chi$ and producing the state estimate $\hat{x}(t)$ is robustly globally exponentially stable 
(RGES) if there exist some 
$C_x,C_w,C_v,r\in \mathbb{R}_{> 0}$ and $\lambda \in \mathbb{R}_{\geq 0}$  such that for any initial conditions $\chi,\hat{\chi} \in \mathcal{X}$, any disturbance sequences~$w\in\mathcal{M_W}$, $v\in\mathcal{M_V}$  and any control input $u\in \mathcal{M_U}$ the following holds for all $t\geq 0$
\begin{align*}
\begin{aligned}
{|}x(t)-\hat{x}(t){|}^r&\leq C_x{|}\chi-\hat\chi{|}^r e^{-\lambda t}+\int_{0}^{t}
C_w{|}w(\tau){|}^re^{-\lambda(t-\tau)} \ d\tau\\&+\sum_{\tau \in [0,t] \cap K_s}
C_v{|}v(\tau){|}^re^{-\lambda(t-\tau)}.
\end{aligned}
\end{align*} 
\end{defn}
 In order to establish robust exponential stability of the estimation error, we consider an exponential version  of sample-based \iioss \color{rev} with quadratic weighting \color{black} as our detectability assumption. 
\begin{assum}[Sample-based exponential \iioss]
\label{ass:eIOSS}
	Consider a set $K$ as in Definition~\ref{def:K}.  
There exist  \allowbreak $P_1,\allowbreak P_2,\allowbreak Q_w,\allowbreak  Q_v,\allowbreak R \succ 0$ and $\eta\in \mathbb{R}_{> 0}$ such that for all $u \in\mathcal{M_U}$, any pairs of disturbance trajectories  $w_1,w_2 \in\mathcal{M_W}$, $v_1,v_2 \in\mathcal{M_V}$ and any pair of initial conditions  $\chi_1$, $\chi_2$ it holds for all $t\geq 0$ and any $K_i\in K$ \par\nopagebreak[2]\predisplaypenalty=0
\begin{align}
		\label{eq:ass1}
\begin{aligned}
{|}\Delta x(t){|}_{P_1}^2&\leq {|}\Delta \chi{|}_{P_2}^2 e^{-\eta t} +\int_{0}^{t}{|}\Delta w(\tau){|}_{Q_w}^2 e^{-\eta(t-\tau)} \ d\tau
\\&+\sum_{\tau \in [0,t] \cap K_i}\big({|}\Delta y(\tau){|}_R^2+|\Delta v(\tau){|}_{Q_v}^2\big)e^{-\eta(t-\tau)}.
\end{aligned}
\end{align}
\end{assum}
\begin{rem}
\label{rem:para}
Note that the weighting matrices $P_2,Q_w,Q_v$ and $R$ and the parameter $\eta$ in the MHE optimization problem (\ref{eq:NLP})-(\ref{eq:objfunc}) correspond to the parameters in Assumption~\ref{ass:eIOSS}.
 This is leveraged to ensure RGES in Theorem~\ref{thm:stab} below. However, this 
  does not restrict any tuning possibilities. If the system is sample-based exponentially \iioss, then the cost function (\ref{eq:objfunc}) can be parameterized using arbitrary positive definite matrices $P_2,Q_w,Q_v$ and $R$ 	\color{rev2} and any $\tilde\eta \in (0,\eta]$ instead of $\eta$. \color{black}This is the case since (\ref{eq:ass1}) can be rescaled accordingly, analogous to the non-sample-based case. 
	However, the choice of $P_2$ influences the minimum horizon length required for the estimator to be RGES. See the discussions in \cite[Remark 3]{Sch24} and \cite[Remark 1]{Sch23}. 
\end{rem}
\color{black}
Under this detectability assumption, we can show that, with a sufficiently large horizon length, the estimation error is guaranteed to be robustly stable. \color{rev} 
Note that, due to our sample-based design, the theoretical analysis largely follows from the DT case in \cite{Kra25d}.
\color{black}
\begin{thm}[sample-based MHE is RGES]
	\label{thm:stab}
Consider a set $K$ as in Definition~\ref{def:K}.  
Let Assumption \ref{ass:eIOSS} hold, assume the measurement sequence $K_s \in K$,  and let the horizon $M\geq  d_{\max}$ be chosen such that $4\lambda^2_{\mathrm{max}}(P_2,P_1)e^{-\eta M}<1$. Then, there exists $\tilde\eta\in \mathbb{R}_{>0}$ such that  the state estimation error of the sample-based  MHE
scheme (\ref{eq:NLP}) satisfies for all $t\geq 0$
\begin{align}
&{|}\hat{e}(t){|}_{P_1}^2\leq  4 \lambda_{\mathrm{max}}(P_2,P_1)\Big(e^{-\tilde\eta t}{|}\hat{e}(0){|}_{P_2}^2
\label{eq:erbound}
\\&+\int_{0}^{t} e^{-\tilde\eta (t-\tau)}  {|}w(\tau){|}_{Q_w}^2 \ d\tau+\sum_{\tau \in [0,t] \cap K_s}
e^{-\tilde\eta (t-\tau)} 2 {|}v(\tau){|}_{Q_v}^2\Big).\nonumber
\end{align}
\end{thm}
\begin{pf}
The MHE scheme presented in Section~\ref{sec:MHE_setting} operates by using a fixed horizon length $M$ whenever new measurements become available and otherwise performing an open-loop prediction. Similar to the DT case in \cite{Kra25d}, it can be shown that this is equivalent to a scheme with varying horizon length (cf. \cite[Proposition 1]{Kra25d}).
Thus, the stability proof in \cite[Theorem 1]{Kra25d} can  be straightforwardly adapted to the CT setting to derive (\ref{eq:erbound}). 
The detailed proof can be found in the appendix.
\hfill $\qed$
\end{pf}
\section{Linear systems and sample-based \iioss}
\label{sec:linSys}
In the following, we will analyze the connections between sample-based observability and sample-based \iioss for linear systems.     As it will be shown in this section,  in the linear case, sample-based exponential \iioss of the system follows from sample-based observability  of the unstable subsystem together with some additional condition on the sampling instances. \color{rev} This condition is required due to the CT setting, which requires additional analysis beyond a straightforward extension of the DT case.
\color{black} We discuss a sufficient condition for this additional condition to hold later in this section (see Theorem~\ref{thm:addcondsat} and the remarks below it).  
Therefore, previously established conditions for linear systems to be sample-based observable (cf. \cite{Wan11,Zen16}) can be utilized to verify or design sampling strategies that satisfy Assumption~\ref{ass:eIOSS} to ensure RGES of the sample-based MHE.
Consider a linear time-invariant system described by
\begin{align}
\begin{aligned}
\dot{x}(t)&=Ax(t)+B u(t)+w(t),\\
y(t)&=C x(t) +D u(t)
\end{aligned}
\label{eq:syslin}
\end{align}
where $A \in \mathbb{R}^{n \times n}$, $B \in \mathbb{R}^{n \times m}$, $C \in \mathbb{R}^{p \times n}$, $D \in \mathbb{R}^{p \times m}$. 
All other assumptions and notation are as in the general nonlinear setup (Section~\ref{sec:Pre}).
Measurement noise s omitted for convenience.
This assumption is made without loss of generality, since (sample-based) \iioss  of system (\ref{eq:syslin}) is equivalent to (sample-based) \iioss of the same system with measurement noise in the output equation, i.e., 	$y(t)=C x(t) +D u(t) +v(t)$. 
\par In the non-sampled case it is well known that \ioss and detectability are equivalent for linear systems \cite{Son97}. For the sake of completeness, we briefly show this for the  integral variant of \ioss that we consider in this paper.
\begin{thm}
\label{thm:detect}
A linear system is detectable if and only if it is
time-discounted exponentially \iioss according to Definition~\ref{def:iiioss}.
\end{thm}
\begin{pf}
	Suppose a linear system is  exponentially \mbox{\iioss}.
	Then  $w, y \equiv 0$ implies $x \rightarrow 0$ as $t\rightarrow \infty$, thus the system is detectable.
	For the opposite direction we use the fact that for a detectable system there exists some matrix $L$ such that $A_L\coloneq A-LC$ is Hurwitz. 
\color{rev}
Hence, there exist $P=P^\top\succ0$ and $Q\succ0$, satisfying
$A_L^\top P + PA_L = -Q$.
Notice that $\Delta \dot{x}(t)=A\Delta x(t) +\Delta w(t) +L(\Delta y(t)-\Delta y(t))
=\Delta \dot{x}(t)=A_L \Delta x(t)+\Delta w(t) +L\Delta y(t)$. Defining the  Lyapunov function $V(\Delta x(t))=\Delta x(t)^\top P\Delta x(t)$, we obtain
\begin{align*}
	\dot V
	&=
	-\Delta x(t)^\top Q\Delta x(t)
	+2\Delta x(t)^\top PL\,\Delta y(t)
	+2\Delta x(t)^\top P\,\Delta w(t).
\end{align*}
Applying Young's inequality, we can derive
\begin{align*}
	\dot V &\leq -c_x|\Delta x(t)|^2 + c_w|\Delta w(t)|^2+ c_y|\Delta y(t)|^2 \\
	&\leq  c_w|\Delta w(t)|^2+ c_y|\Delta y(t)|^2 
	\end{align*}
	for some $c_x,c_w,c_y>0$. Integrating over $[0,t]$ yields
\begin{align*}
V(t)-V(0)
\leq
\int_0^t \big( c_w|\Delta w(\tau)|^2+c_y|\Delta y(\tau)|^2\big)\,d\tau.
\end{align*}
Hence, 
\begingroup 
\allowdisplaybreaks[0] 
\begin{align}
 |\Delta x(t)|_P^2
\leq
|\Delta \chi|_P^2
+
\int_0^t \big( c_w|\Delta w(\tau)|^2+c_y|\Delta y(\tau)|^2\big)\,d\tau.	\label{eq:eIOSSsq} \\ \noindent\hfill \qed \nonumber
\end{align}
\endgroup
\end{pf}
The following definition will be useful in this section.
\begin{defn}[Sample-based observability matrix \cite{Kra25b}]
\label{def:SOm}
Consider a set of arbitrary time instances $\{ \tau_i \}_{i=1}^k$ for some $k\geq 1$, with $0 \leq \tau_1$ and $\tau_i < \tau_{i+1}$ for all $i=1,...,k-1$.
The sample-based observability matrix  is given by
\begin{align}
	\label{eq:Osc}
O_{\mathrm{s}}(A,C)=\begin{pmatrix}
(Ce^{A\tau_1})^\top&(Ce^{A\tau_2})^\top& \ldots&(Ce^{A\tau_k})^\top 
\end{pmatrix}^\top.
\end{align}
\end{defn}
In \cite[Definition~2]{Kra25b}, it was discussed that if the sample-based observability matrix 	$O_{\mathrm{s}}(A,C)$ has full column rank, one can uniquely reconstruct the initial state from knowledge of the input sequence (in case of zero disturbances) and sampled  outputs, meaning that the system is sample-based observable.
\par 	Before stating the first main result of this section, we propose the following lemma, that will be used in the proof of Theorem~\ref{thm:lin1} below.
\begin{lem}
\label{lem:singval_min}
If $\mathrm{rank}(O_{\mathrm{s}}(A,C))=n$ with $\{ \tau_{i,t} \}_{i=1}^{k_t} \subseteq [t-T,t]$ and $\sigma_{\mathrm{min}}(O_{\mathrm{s}}(A,C)e^{-A(t-T)})\geq \bar{\sigma}$  for some $\bar\sigma >0$, then there exist some matrices $\mu_i\in\mathbb{R}^{p\times p}$, $i=1,\ldots,k_t$ and some $\bar{\mu}>0$ such that
\begin{align}Ce^{AT}=M O_{\mathrm{s}}(A,C)e^{-A(t-T)}\label{eq:lem}\end{align}  with $M=\begin{pmatrix}\mu_1 &  \mu_2 & \cdots & \mu_{k_t}\end{pmatrix} \in \mathbb{R}^{p \times k_tp}$ and ${|}\mu_i{|}\leq \bar{\mu}$, $i=1,\ldots, k_t$.
\end{lem}
\begin{pf}
One possible choice for matrix $M$ which satisfies $Ce^{AT}=M O_{\mathrm{s}}(A,C)e^{-A(t-T)}$ is given by 
$M=Ce^{AT}\tilde{O}^+$
with $\tilde{O}=O_{\mathrm{s}}(A,C)e^{-A(t-T)}$ and $\tilde{O}^+$ being the pseudoinverse of $\tilde{O}$, i.e., $\tilde{O}^+=(\tilde{O}^\top \tilde{O})^{-1} \tilde{O}^\top$.
Since ${|}\tilde{O}^+{|}=\frac{1}{\sigma_{\mathrm{min}}(\tilde{O})}$,
\begin{align*}
|M|\leq \tfrac{{|}Ce^{AT}{|}}{\sigma_{\mathrm{min}}(\tilde{O})}\leq  \tfrac{{|}Ce^{AT}{|}}{\bar{\sigma}}.
\end{align*}
Thus, (\ref{eq:lem}) holds with $\bar\mu=\frac{{|}Ce^{AT}{|}}{\bar{\sigma}}$.
$\qed$
\end{pf}
In the following theorem we exploit the sufficient condition from Theorem~\ref{thm:sufcond}. For this  note the following.
\begin{rem}
\label{rem:sufconlin}
If the functions ${\alpha},\alpha_x,\gamma_1, \gamma_2,\gamma_3$ ($\bar{\alpha},\bar\alpha_x,\bar\gamma_1, \bar\gamma_2,\bar\gamma_3$) in Definition~\ref{def:iiioss} (Definition~\ref{def:sbIOSS}) are linear or quadratic, meaning that, for example $\alpha$ takes the form
\begin{align}
\alpha(|x|)=c|x| \quad \mathrm{or} \quad  \alpha(|x|)=x^\top P x
\label{eq:linquad}
\end{align}
with $c\in \mathbb{R}_{>0}$ and $P\succ0$, then we have an exponential variant of (sample-based) \iioss. By close inspection  of the proof of Theorem~\ref{thm:sufcond} it can be seen that (\ref{eq:suffcond}) with $\alpha_w,\alpha_y,\alpha_v$ being linear or quadratic functions as in (\ref{eq:linquad}) is a sufficient condition for sample-based exponential \iioss of system (\ref{eq:syslin}) with  linear or quadratic functions, respectively. Here, it is also used that for a linear system Assumption~\ref{ass:f} holds with $\rho$ being a linear function.
\end{rem}	
 While in theory an arbitrary number of measurements could occur within a finite time interval of length $T$,  we make the reasonable assumption in the following theorem that there exists a finite upper bound $\kappa(T)$ on the number of measurement instances over any such interval. Now we can establish the following results on sample-based \iioss of a sample-based observable system.
Besides requiring that the sample-based observability matrix has full column rank, which follows from sample-based observability, we assume a lower bound on the smallest singular value of the sample-based observability matrix. 
This condition, for example, may not be satisfied if the measurement instances are arbitrarily close to each other. 
\color{rev}
For instance, consider $n$ measurement instances of a single-output system that yield a full rank sample-based observability matrix. Clearly, the smallest singular value of $O_{\text{s}}(A,C)$ is arbitrarily close to zero if any  two consecutive measurement times, $t_i$ and $t_{i+1}$, become arbitrarily close.
\color{black}
\begin{thm}
\label{thm:lin1}
Consider a set $K$ as in Definition~\ref{def:K} with the number of measurements in $K_1 \cap [t-T,t]$ being less than or equal to some finite  $\kappa(T)\in \mathbb{R}_{>0}$ for any $T\in  \mathbb{R}_{>0}$ and all $t>T$.
If there exist some finite $T$ and some $\bar{\sigma}>0$ such that the sample-based observability matrix with time indices in the set   $K_1 \cap [t-T,t]$ has full column rank and satisfies $\sigma_{\mathrm{min}}(O_{\mathrm{s}}(A,C)e^{-A(t-T)})\geq\bar{\sigma}$ for all $t\geq T$, then system (\ref{eq:syslin}) is sample-based exponentially \iioss with respect to $K$.
\end{thm}
\begin{pf}
Since the sample-based observability matrix has full column rank, the system is sample-based observable, and therefore also observable. This in turn implies by Theorem~\ref{thm:detect} that the system is exponentially \iioss. Moreover, Assumption~\ref{ass:f} is always satisfied for linear systems. Thus, it only remains to be shown that condition (\ref{eq:suffcond}) in  Theorem~\ref{thm:sufcond} holds with linear functions,  implying sample-based exponential \iioss by Remark~\ref{rem:sufconlin}.
Consider the output  difference at some time $t=\vartheta T+l$ with $\vartheta\in \mathbb{I}_{>0}$ and  $l \in \mathbb{R}_{\geq 0}$ 
\begin{align}
\begin{aligned}
\Delta y(t)=Ce^{AT}\Delta x(t-T)+\int_{t-T}^{t} Ce^{A(t-\tau)}\Delta w(\tau) \ d\tau.
\end{aligned}
\label{eq:Deltay}
\end{align}	
Thus, we can write 
\begin{align}
	\begin{aligned}
{|}	\Delta y(t){|}\leq&{|}Ce^{AT}\Delta x(t-T){|}\\&+\int_{t-T}^{t} {|}Ce^{A(t-\tau)}\Delta w(\tau){|} \ d\tau.
\end{aligned}
\label{eq:Deltayabs}
\end{align}
By Lemma~\ref{lem:singval_min} we know that 
\begin{align}
	\begin{aligned}
&{|}Ce^{AT}\Delta x(t-T){|}= {|}M O_{\mathrm{s}}(A,C)e^{-A(t-T)}\Delta x(t-T){|}\\
\leq   &\sum_{j\in K_1 \cap [t-T,t]} \bar\mu{|}Ce^{A(j-t+T)}\Delta x(t-T){|}. 
\end{aligned}
\label{eq:boundFromLem}
\end{align}
Using (\ref{eq:Deltay}),  we can further upper bound (\ref{eq:boundFromLem}) as follows
\begin{align}
\hspace{-0.8em}
\begin{aligned}
&{|}Ce^{AT}\Delta x(t-T)| \\
\leq& \bar\mu  \sum_{j\in K_1 \cap [t-T,t]} \Big({|}\Delta y(j){|}+\int_{t-T}^{j}{|}Ce^{A(j-\tau)}\Delta w(\tau){|} \ d\tau \Big). 
\end{aligned}
\label{eq:boundT}
\end{align}
Substituting (\ref{eq:boundT}) in (\ref{eq:Deltayabs}) we obtain
\begin{align*}
{|}\Delta y(t){|}\leq \bar\mu  \sum_{\mathclap{j\in K_1 \cap [t-T,t]}} \hspace{3mm} {|}\Delta y(j){|}+(\bar\mu \kappa(T)+1)\int_{t-T}^{t}k_w \ {|}\Delta w(\tau){|} \ d\tau
\end{align*}
with $k_w \geq {|}Ce^{Ai}{|}$ for all $i\in [0,T]$. 
\color{rev}
To derive (\ref{eq:suffcond}) with quadratic functions, we square both sides and apply Young’s and Jensen’s inequalities to obtain the following upper bound
\begin{align}
	\begin{aligned}
		{|}\Delta y(t){|}^2\leq& 2\bar\mu^2 \kappa(T)  \sum_{j\in K_1 \cap [t-T,t]} {|}\Delta y(j){|}^2\\&+2(\bar\mu \kappa(T)+1)^2T\int_{t-T}^{t}k_w^2 \ {|}\Delta w(\tau){|}^2 \ d\tau.
	\end{aligned}
	\label{eq:boundy_lin}
\end{align}
\color{black}
Notice that  $\int_T^t {|}\Delta y(j) {|}^2 \  dj=\sum_{i=1}^{\vartheta-1} \int_{iT}^{(i+1)T}{|}\Delta y(j) {|}^2 \ dj + \int_{\vartheta T}^t{|}\Delta y(j) {|}^2 \ dj$ and $\int_{iT}^{(i+1)T}{|}\Delta y(j) {|}^2 \ dj  \leq T \|\Delta y \|^2_{iT:(i+1)T}$. Applying (\ref{eq:boundy_lin}) for each of these $\vartheta$ time intervals,
we can derive the following bound
\begin{align}
\int_{T}^tc_2 {|}\Delta y(\tau){|}^2 \  d\tau
&\leq \int_0^t4T^2 c_2(\bar\mu \kappa(T)+1)^2k_w^2{|} \Delta w(\tau)
{|}^2 \ d\tau\nonumber \\&+ \sum_{\tau \in [0,t] \cap K_1}4T \bar\mu^2\kappa(T)  c_2 {|}\Delta y(\tau){|}^2
\label{eq:sufconK1}
\end{align}
where $c_2\in\mathbb{R}_{>0}$ is such that $\gamma_2(s)=c_2 s^2$ satisfies the exponential \iioss bound (\ref{eq:ioss}) for a linear system (compare proof of Theorem~\ref{thm:detect}). 
Note that when applying (\ref{eq:boundy_lin}) to the essential supremum of $|\Delta y(t)|^2$  in each of the $\vartheta$ subintervals, the resulting bounds involve the past interval $[t-T,t]$ which may extend  into the preceding subinterval. Hence, due to splitting the complete  interval $[T,t]$ in $T$-long subintervals, the resulting bounds from applying (\ref{eq:boundy_lin})  can overlap only between neighbouring subintervals.  This overlap accounts for the additional factor  two in the bound  (\ref{eq:sufconK1}).
Since we assume sample-based observability in any $T$-long interval, (\ref{eq:sufconK1}) also holds for all $K_i\in K$ by definition.
Thus (\ref{eq:suffcond}) holds with $\alpha_w(s)=4T^2 c_2(\bar\mu \kappa(T)+1)^2k_w^2 s$, $\alpha_y(s)=4T \bar\mu^2\kappa(T)s$ and $t^*=T$. 
\hfill $\qed$
\end{pf}
\begin{rem}
	\color{rev}
	Note that the proofs of Theorem~\ref{thm:detect} and Theorem~\ref{thm:lin1} are formulated to yield quadratic functions in (\ref{eq:eIOSSsq}) and (\ref{eq:sufconK1}) in order to obtain sample-based exponential \iioss in the form of Assumption~\ref{ass:eIOSS}. However, they can also be carried out to yield linear functions in (\ref{eq:eIOSSsq}) and (\ref{eq:sufconK1})  instead. \color{black}
\end{rem}
 The result in Theorem~\ref{thm:lin1} can be relaxed. In fact, it is sufficient that the conditions in Theorem~\ref{thm:lin1} only hold for the part of the system that is not asymptotically stable. 	For this, consider the system in Jordan canonical form with
\begin{align*}
A_{\mathrm{J}}=T_{\mathrm{J}}^{-1}AT_{\mathrm{J}}=\begin{pmatrix}
A_{\mathrm{s}}&0\\0&A_{\mathrm{us}} \end{pmatrix}, \quad C_{\mathrm{J}}=CT_{\mathrm{J}}=\begin{pmatrix} C_{\mathrm{s}}&C_{\mathrm{us}}\end{pmatrix}
\end{align*}
where all the eigenvalues of $A$  with negative real part are included in $A_{\mathrm{s}}$, and the remaining ones in $A_{\mathrm{us}}$.
\begin{thm}
\label{thm:sbobsvusioss}
\looseness=-1
Consider a set $K$ as in Definition~\ref{def:K} with the number of measurements in $K_1 \cap [t-T,t]$ being less than or equal to some finite  $\kappa(T)\in \mathbb{R}_{>0}$ for any $T\in  \mathbb{R}_{>0}$ and all $t>T$.
If there exists some finite $T$ and some $\bar{\sigma}>0$ such that the sample-based observability matrix $O_{\mathrm{s}}(A_{\mathrm{us}},C_{\mathrm{us}})$ with time indices in the set   $K_1 \cap [t-T,t]$ has full column rank and  $\sigma_{\mathrm{min}}(O_{\mathrm{s}}(A_{\mathrm{us}},C_{\mathrm{us}})e^{-A_{\mathrm{us}}(t-T)})\geq\bar{\sigma}$ for all $t>T$, then system (\ref{eq:syslin}) is sample-based exponentially \iioss with respect to $K$.
\end{thm}
\begin{pf}	
The proof can, for the most part, be carried out analogously to the proof of the corresponding theorem in the DT case  \cite[Theorem 4]{Kra25d}.
The main idea is to use the  \iiss bound for the  stable subsystem and the sample-based  \iioss bound for the unstable subsystem  resulting from application of Theorem~\ref{thm:lin1}. 
Note that we use in the DT case in \cite{Kra25d} a max-based formulation of the \ioss bound. 
Due to the time-discounting, we can transform the sum term of the sampled outputs in the sample-based \iioss bound (\ref{eq:sbiIOSS}) to a max-based formulation (see e.g. \cite[Equation (25)]{Knu20}). This allows us to follow similar steps as in the DT case to derive a sample-based  \iioss bound for system (\ref{eq:syslin}).
\hfill$\qed$
\end{pf}	
For CT observable systems, sampling strategies that ensure (\ref{eq:Osc}) has full column rank are shown in \cite{Wan11,Zen16}. 
These works showed that (\ref{eq:Osc}) is guaranteed to have full column rank if more than $k^*$ measurements are taken within an interval of length $T$, where
\begin{equation}\label{eq:sbObsvC}
	k^{*}= r-1+\tfrac{T\delta_\lambda}{2\pi}
\end{equation}	
with $r=\sum_{i=1}^\varsigma r_i$, $r_i$ denoting the index of the eigenvalue $\lambda_i$ of the matrix $A$, $\varsigma$ being the number of pairwise distinct eigenvalues of the system and $\delta_\lambda=\max_{1\leq i,j\leq \varsigma} \Im (\lambda_i-\lambda_j)$ with  $\Im(\cdot)$ indicating the imaginary part \cite{Zen16}.
In the following theorem, we show that selecting  $k>k^*$ samples in a $T$-long time interval that additionally satisfy  $\tau_{i+1}-\tau_i\geq\epsilon$, $i=1,\ldots, k-1$ for some $\epsilon > 0$ implies that there exists a lower bound on the minimum singular value of the sample-based observability matrix.
\begin{thm}
\label{thm:addcondsat}
Consider $T>0$ and define $k^*$ as in (\ref{eq:sbObsvC}).
 Fix $\epsilon>0$ and choose $k=\min\{i|i>k^*, i\in \mathbb{I}\}$ sampling times
\[
\tau_1,\dots,\tau_k\in[t-T,t],\qquad \tau_{i+1}-\tau_i\geq\epsilon\quad\forall i=1,\ldots, k-1.
\]
Then, there exists $\bar\sigma > 0$ such that for all $t \geq T$ the  sample-based observability matrix in the time interval $[t-T,t]$, referred to as $O_{\mathrm{s},t}(A,C)$, satisfies 
\mbox{$\sigma_{\min}(	O_{\mathrm{s},t}(A,C)e^{-A(t-T)})\geq \bar{\sigma}$}. 
\end{thm}
\begin{pf}
We denote the set of exponents of the matrix exponentials in $O_{\mathrm{s,t}}(A,C)e^{-A(t-T)}$ by $\mathcal{T}=\Bigl\{(s_1,s_2,\dots,s_k)|s_1\in[0,T-(k-1)\epsilon],\ s_i\in[s_{i-1}+\epsilon,T-(k-i)\epsilon], \ \forall i=2,\ldots,k\Bigr\}$ where $s_i=\tau_i-t+T,\ i=1,\ldots,k$.
Note that $\mathcal{T}$ is a compact set and 	$g: \mathcal{T}\to\mathbb{R}$,	
$g(s_1,s_2,\ldots,s_k)=\sigma_{\min}\Bigl(\begin{pmatrix}
(Ce^{As_1})^\top&(Ce^{As_2})^\top& \ldots&(Ce^{As_k})^\top 
\end{pmatrix}^\top\Bigr)$
is a continuous function. Since $k>k^*$, the results from \cite{Zen16}, discussed above, guarantee that $O_{\mathrm{s},t}(A,C)e^{-A(t-T)}$ has full column rank. Hence, on the compact set $\mathcal{T}$, the continuous function $g$ never vanishes.
Therefore, by the extreme‐value theorem $g$ attains a strictly positive minimum
	$
\bar{\sigma}=\min_{(s_1,s_2,\ldots,s_k)\in\mathcal{T}}
g(s_1,\dots,s_k)
>0.
$
Thus, for every admissible choice of
$\{\tau_i\}_{i=1}^k$ one has $\sigma_{\min}\Big(O_{\mathrm{s},t}(A,C)e^{-A(t-T)}\Big)  \geq \bar{\sigma}$. 
\hfill $\qed$
\end{pf}
 Hence, applying (\ref{eq:sbObsvC}) to the subsystem $(A_{\mathrm{us}},C_{\mathrm{us}})$  and using a measurement sequence with  more than $k^*$ samples in every  $T$-long time interval, with consecutive sampling instants satisfying  $\tau_{i+1}-\tau_i\geq\epsilon$,  sample-based exponential \iioss of (\ref{eq:syslin}) follows from Theorem~\ref{thm:sbobsvusioss}.
\section{Simulation example}
\label{sec:Sim}
To illustrate the functionality of the sample-based MHE scheme from Section \ref{sec:sbMHE}, we use the highly nonlinear  six-dimensional model of the pituitary-thyroid feedback loop  \color{rev2} provided in the supplementary material\footnote{ \color{rev2} \url{https://www.frontiersin.org/journals/endocrinology/articles/10.3389/fendo.2022.882788/full#supplementaryMaterial}} of \cite{Wol22} for a hypothyroid patient. \color{black}  In such biomedical applications, obtaining measurements typically requires
taking blood samples, which is impractical to do on a frequent basis. \color{black}
The system state is $x=\begin{pmatrix} T_{4,th}& T_4& T_{3,p}& T_{3,c}& TSH & TSH_c\end{pmatrix}^\top$ where $T_4,T_{3,p}$, and $TSH$ are the peripheral concentrations of the hormones thyroxine ($T_4$), triiodothyronine ($T_3$), and the thyroid-stimulating
hormone ($TSH$) and $T_{4,th}, T_{3,c},TSH_c$ are  the internal concentrations of $T_4$ in the thyroid and  of $T_3,TSH$ in the pituitary, respectively. 
We consider the peripheral concentrations as the system's output that we can directly measure by taking blood samples. Furthermore we assume additive process noise $w\in \mathbb{R}^2$ acting on the first and sixth state equations,  \color{rev2} namely  $|w_{T_{4,th}}|\leq 10^{-13}$ and  $|w_{TSH,c}|\leq 0.6$, respectively, \color{black} and  additive measurement noise $v\in \mathbb{R}^3$ with  $|v_1|\leq 5\cdot 10^{-10}$,  $|v_2|\leq 3\cdot 10^{-11}$ and  $|v_3|\leq 0.05$. 
The true initial state is  $x(0)=\begin{pmatrix}4.4\cdot 10^{-13}& 2.1\cdot 10^{-8} &9.6 \cdot 10^{-10}& 2.1\cdot 10^{-9}&4.5& 4.8 \end{pmatrix}^\top$ and as the  prior estimate we use  \emergencystretch=5em  $\bar{x}(0)= \begin{pmatrix}7.4\cdot 10^{-13}& 1.4\cdot 10^{-8} &4.2\cdot 10^{-10}& 9\cdot 10^{-10}&7.7& 8.2\end{pmatrix}^\top$.
\begin{figure}[!t]
\centering
\setlength\figurewidth{0.9\columnwidth} 
 \begingroup
\input{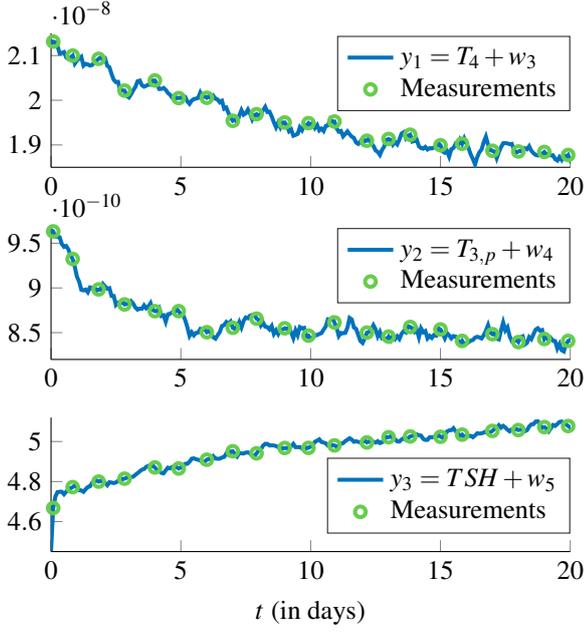}
\endgroup
\caption{Noisy system output $y$ (blue) and infrequently sampled discrete measurements (green) used in the sample-based MHE scheme.} 
\label{fig:hypo_y}
\end{figure}
\color{rev2} We consider an irregular measurement sequence in which, on average, one measurement is obtained every 24 hours, with the measurement time varying within an eight-hour window each day. 
We empirically verified sample-based exponential i-iIOSS by evaluating (\ref{eq:ass1}) along over  $10^6$ pairs of simulated trajectories, considering a relevant range of initial conditions around those specified above that covers the expected range of state values\footnote{$
	\chi\in
	[2,8]\cdot 10^{-13}
	\times [1,5]\cdot 10^{-8}
	\times [4,15]\cdot 10^{-10}
	\times [0.5,5]\cdot 10^{-9}
	\times [4,8] \times [4,9]$}, and noise trajectories satisfying the bounds given above.
\color{black} For the presented setup we applied the sample-based MHE scheme
from Section \ref{sec:sbMHE}. 
%
In Figure~\ref{fig:hypo_y}, the noisy output signal $y(t)$  is displayed and the time instants  at which measurements were taken are marked.  For the parameterization of the cost function we selected $\eta=0.35$, $P_2=\mathrm{diag}\{1,5,1,1,1,1\}$, \color{rev2}$Q_w=1.39\cdot10^{-3}\mathrm{diag}\{1,0.1\}$, \color{black}        $Q_v=\mathrm{diag}\{1000,1000,100\}$ and $R=\mathrm{diag}\{200,150,100\}$. With this choice and a horizon length of $M=5\ \textrm{days}$, we obtain accurate state estimates that quickly converge close to the true states, as can be seen in Figures~\ref{fig:hypo_xp} and \ref{fig:hypo_xc}.
\color{rev}
While the use of a standard moving horizon estimator is not practically feasible for the considered application due to the limited availability of measurements, we nevertheless performed simulations using a standard MHE\footnote{\color{rev}Here, the output is assumed to be continuously available, and the optimization problem is solved at each sampling instant. In our implementation, a sampling time of 2 hours was used.} for comparison purposes.
\color{black} Figure~\ref{fig:hypo_xp} shows the true hormone concentrations, MHE estimates, and sample-based MHE estimates  of the hormone concentrations that we can directly measure and Figure~\ref{fig:hypo_xc} shows the unmeasured hormone concentrations and their MHE and sample-based MHE estimates. 
\color{rev}
 It can be seen that the MHE estimates converge slightly faster to the true states than the sample-based MHE estimates, but the overall results are very similar. Hence, despite relying on considerably less measurement data and solving the NLP less often, we obtain similar estimation accuracy.
\color{black}
   \color{rev}  Furthermore, we performed 50 simulations with different noise realizations, which support this observation. To account for large differences in state magnitudes and mitigate ill-conditioning, the simulations were carried out in a transformed state space obtained via $T=\mathrm{diag}\{10^{12}, 10^{7},  10^{9}, 10^{8}, 1, 1\}$ (cf. the scaled ODEs in \cite{Sir26}). In this representation, the corresponding average root mean squared error yields  0.2244 for the standard MHE and 0.2275 for the sample-based MHE.
  \color{black}
\begin{figure}[!t]
\centering
\setlength\figurewidth{0.9\columnwidth} 
\input{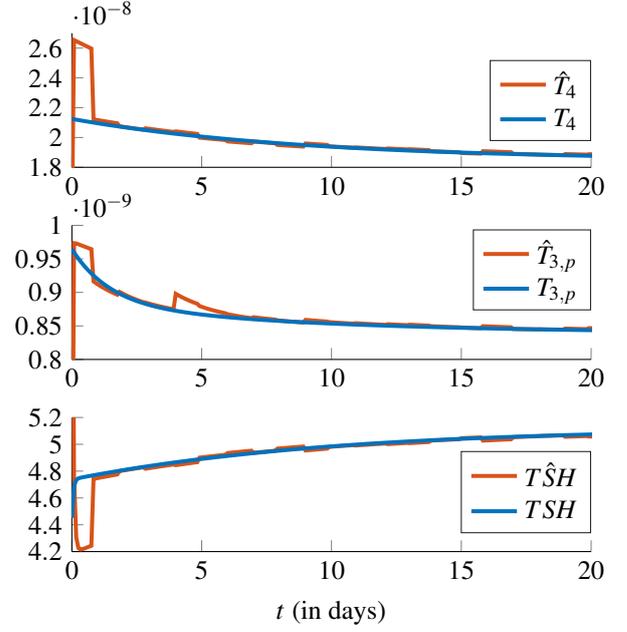}
\caption{Real system states (blue), sample-based MHE  estimates (red), and MHE estimates (green) of the hormone concentrations that are directly measurable.} 
\label{fig:hypo_xp}
\end{figure}
\begin{figure}[!t]
\centering
\setlength\figurewidth{0.9\columnwidth} 
\input{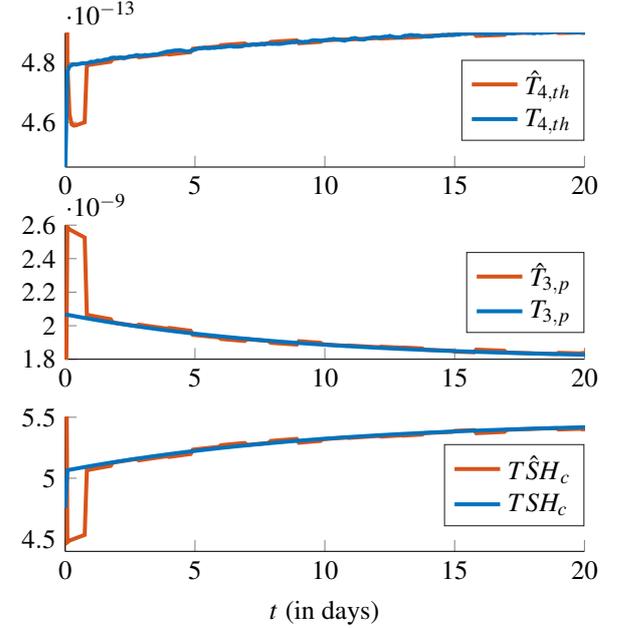}
\caption{Real system states (blue), sample-based MHE estimates (red), and MHE estimates (green)  of the unmeasured  hormone concentrations.} 
\label{fig:hypo_xc}
\end{figure}
\section{Conclusion}
In this work, we studied detectability and state estimation of CT systems in case of irregular measurement sequences. We introduced a sample-based version of \iioss  that considers 
only the limited output information available from the irregular samples of the system. Additionally, we provided a sufficient condition  for an \iioss system to be sample-based \iioss. This condition was exploited to connect the proposed property of sample-based \iioss with previously established results on sample-based observability for linear systems. Furthermore, we proposed a robustly stable sample-based MHE scheme. The applicability of this method was illustrated by a simulation  example of a biomedical application, specifically the estimation of hormone concentrations of the pituitary-thyroid feedback loop, where typically only few and irregular measurements are available.
\begin{ack}                               
This work received funding from the European Research Council (ERC) under the European Union’s Horizon 2020 research and innovation programme (grant agreement No 948679).
\end{ack}
\bibliographystyle{ieeetr}
\bibliography{autosam}           
\appendix
\section{Proof of Theorem~\ref{thm:stab}}
As stated in the short proof below Theorem~\ref{thm:stab}, it can be shown that the proposed sample-based MHE scheme is equivalent, i.e., provides the same estimates, to a scheme with the
varying horizon length $M_{t}=\max\{t,M+\delta_t\}$ where $\delta_t:=t-\max\{0,j \in K_s|j\leq t\}$ refers to the time that has passed since the last measurement was available. Note that if $t\in K_s$, then $\delta_t=0$.
Solving the optimization problem at a high sampling frequency that includes the (generally sparse) measurement instants, and using this new definition of $M_t$,  the solution to the optimization problem is given for any $t\geq0$ by 
\begin{align*}
	\begin{aligned}	
		\hat\chi_t^*&=\hat\chi_{t-\delta_t}^*,\\
		\hat{w}_t^*(j) &= \hat{w}_{t-\delta_t}^*(j), \ j\in [0,M_{t-\delta_t}],\\
		\hat{w}_t^*(j) &= 0, \ j \in  (M_{t-\delta_t},M_t]. 
	\end{aligned}
\end{align*}
This result follows by adapting  Proposition~1 in \cite{Kra25d} to the cost function considered  here.
That is why it is sufficient to solve the NLP (\ref{eq:NLP})-(\ref{eq:objfunc}) only at $t\in K_s$, and obtaining $\hat{x}(t)$  in between two instances  via  open-loop  prediction, as formulated in Section~\ref{sec:MHE_setting}.

In the following we now give the stability proof for the MHE formulation with $M_{t}=\max\{t,M+\delta_t\}$ which thus implies that the derived stability result also holds for the sample-based MHE formulation in Section~\ref{sec:MHE_setting}.
\par Due to the NLP constraints, the estimated trajectories satisfy (1), $\hat{x}_t(\tau)\in \mathcal{X}$ for all $\tau\in [t-M_t,t]$ and $\hat{w}_t(\tau)\in \mathcal{W},\hat{y}_t(\tau)\in \mathcal{Y}$ for all $\tau\in [t-M_t,t]$. Furthermore, we define $\zeta_{\tau}\coloneqq\min\{j|j\in [\tau,\infty)\cap K_s\}-\tau$, i.e., $\zeta_\tau$ is the  time  that will pass until the next measurement from $\tau$ on (and $\zeta_{\tau}=0$ in case that $\tau \in K_s$).
\par Consider some time $t\geq M$. Since $M\geq d_{\text{max}}$, we can  partition the time interval $[t-M_t,t]$ into $[t-M_t,t-M_t+\zeta_{t-M_t}]$ and $[t-M_t+\zeta_{t-M_t},t]$. 
Notice that the second interval is useful because, by definition of $\zeta_{\tau}$, there is a measurement at time $t - M_t + \zeta_{t-M_t}$; hence, every subsequent measurement instance within the second time interval follows the pattern described in Definition~2 using a contiguous subsequence of $\mathcal{D}$.
Thus, we can apply  (\ref{eq:ass1})  taking $t-M_t+\zeta_{t-M_t}$ as initial time,  which yields
\begin{align}
		&{|}\hat{x}(t)-x(t){|}_{P_1}^2 \nonumber\\\leq& e^{-\eta(M_t-\zeta_{t-M_t})} {|}\hat{x}_t^*(\zeta_{t-M_t})-x(t-M_t+\zeta_{t-M_t}){|}_{P_2}^2\nonumber\\&+\int_{\zeta_{t-M_t}}^{M_t}e^{-\eta(M_t-\tau)}{|}\hat{w}_t^*(\tau)-w(t-M_t+\tau){|}_{Q_w}^2 d\tau\nonumber \\& 
		+\sum_{\substack{\tau+t-M_t+\zeta_{t-M_t}\\\in [t-M_t+\zeta_{t-M_t},t] \cap K_s}}e^{-\eta(M_t-\tau)}\Big({|}\hat{v}_t^*(\tau)-v(t-M_t+\tau){|}_{Q_v}^2\nonumber\\&+{|}\hat{y}_t^*(\tau)-y(t-M_t+\tau){|}_R^2\Big). 	\label{eq:eiossEst}
\end{align}
Due to $\zeta_{t-M_t}< d_{\text{max}}$ 
and Assumption~7 considering all $K_i\in K$, we can conclude that for a $\zeta_{t-M_t}$-long time interval it holds that
\begin{align}
		&	{|}\hat{x}_t^*(\zeta_{t-M_t})-x(t-M_t+\zeta_{t-M_t}){|}_{P_1}^2 \nonumber\\\leq&e^{-\eta(\zeta_{t-M_t})} {|}\hat{\chi}_t^*-x(t-M_t){|}_{P_2}^2 	\label{eq:eiossEst_eps}
		\\&+\int_{0}^{\zeta_{t-M_t}}e^{-\eta(\zeta_{t-M_t}-\tau)}{|}\hat{w}_t^*(\tau)-w(t-M_t+\tau){|}_{Q_w}^2 d\tau.\nonumber
\end{align}
Using the fact that
\begin{align*}
	\begin{aligned}
		&{|}\hat{x}_t^*(\zeta_{t-M_t})-x(t-M_t+\zeta_{t-M_t}){|}_{P_2}^2\\\leq&\lambda_{\text{max}}(P_2,P_1){|}\hat{x}_t^*(\zeta_{t-M_t})-x(t-M_t+\zeta_{t-M_t}){|}_{P_1}^2
	\end{aligned}
\end{align*} 
and combining (\ref{eq:eiossEst}) and (\ref{eq:eiossEst_eps}) we can write
\begin{align}
		&	{|}\hat{e}(t){|}_{P_1}^2=	{|}\hat{x}(t)-x(t){|}_{P_1}^2 \nonumber\\\leq& \lambda_{\text{max}}(P_2,P_1) \big(e^{-\eta M_t} {|}\hat{\chi}_t^*-x(t-M_t){|}_{P_2}^2\nonumber\\&+\int_{0}^{M_t}e^{-\eta(M_t-\tau)}{|}\hat{w}_t^*(\tau)-w(t-M_t+\tau){|}_{Q_w}^2 	\label{eq:eiossEstcomb}\\& 
		+\sum_{\tau+t-M_t\in [t-M_t,t] \cap K_s}e^{-\eta(M_t-\tau)}\Big({|}\hat{v}_t^*(\tau)-v(t-M_t+\tau){|}_{Q_v}^2\nonumber\\&+{|}\hat{y}_t^*(\tau)-y(t-M_t+\tau){|}_R^2\Big).\nonumber
\end{align}
Note that for $t<M$, (\ref{eq:ass1}) is directly applicable, and thus (\ref{eq:eiossEstcomb}) is a valid bound for all $t\geq 0$. 
Since 
\begin{align*}
	\begin{aligned}
		{|}\hat{w}_t^*(\tau)-w(t-M_t+\tau){|}_{Q_w}^2\leq& 2{|}\hat{w}_t^*(\tau){|}_{Q_w}^2+ 2{|}w(t-M_t+\tau){|}_{Q_w}^2,
		\\
		{|}\hat{v}_t^*(\tau)-v(t-M_t+\tau){|}_{Q_v}^2\leq& 2{|}\hat{v}_t^*(\tau){|}_{Q_v}^2+ 2{|}v(t-M_t+\tau){|}_{Q_v}^2
	\end{aligned}
\end{align*}
and
\begin{align*}
	\begin{aligned}
		&{|}\hat{\chi}_t^*-x(t-M_t){|}_{P_2}^2\\=& {|}\hat{x}(t-M_t)-x(t-M_t)+\hat{\chi}_t^*-\hat{x}(t-M_t){|}_{P_2}^2\\
		\leq &2{|}\hat{x}(t-M_t)-x(t-M_t){|}_{P_2}^2+2{|}\hat{\chi}_t^*-\hat{x}(t-M_t){|}_{P_2}^2,
	\end{aligned}
\end{align*}
we obtain%
\begin{align*}
	\begin{aligned}
		{|}\hat{e}(t){|}_{P_1}^2 \leq& \lambda_{\text{max}}(P_2,P_1)\Big(2 e^{-\eta M_t} {|}\hat{x}(t-M_t)-x(t-M_t){|}_{P_2}^2\\&+\int_{t-M_t}^{t}2e^{-\eta (t-\tau)} {|}w(\tau){|}_{Q_w}^2 d \tau \\&+ \sum_{\tau \in [t-M_t,t]\cap K_s}2 e^{-\eta (t-\tau)} {|}v(\tau){|}_{Q_v}^2 \\&+J(\hat{\chi}_t^*, \hat{w}_t^*,\hat{v}_t^*, \hat{y}_t^*,t)\Big).
	\end{aligned}
\end{align*}
By optimality, we have  $J(\hat{\chi}_t^*, \hat{w}_t^*,\hat{v}_t^*, \hat{y}_t^*,t) \leq J(x(t-M_t), w_t,v_t, y_t,t)$ 	with  $w_t$ and  $v_t$ referring to the true disturbance trajectories on the interval $[t-M_t, t]$. Thus,
\begin{align}
		{|}\hat{e}(t){|}_{P_1}^2 \leq& \lambda_{\text{max}}(P_2,P_1)\Big(4e^{-\eta M_t} {|}\hat{x}(t-M_t)-x(t-M_t){|}_{P_2}^2\nonumber
		\\&+4\int_{t-M_t}^{t}e^{-\eta (t-\tau)} {|}w(\tau){|}_{Q_w}^2 d \tau 	\label{eq:boundMtstep1}
		\\&+ 4\sum_{\tau \in [t-M_t,t]\cap K_s} e^{-\eta (t-\tau)} {|}v(\tau){|}_{Q_v}^2 \Big).\nonumber
	\end{align} 
	Since ${|}\hat{x}(t-M_t)-x(t-M_t){|}_{P_2}^2\leq\lambda_{\text{max}}(P_2,P_1){|}\hat{x}(t-M_t)-x(t-M_t){|}_{P_1}^2$, 
	\begin{align*}
		\begin{aligned}
			{|}\hat{e}(t){|}_{P_1}^2  \leq &4 e^{-\eta M_t} \lambda^2_{\text{max}}(P_2,P_1){|}\hat{x}(t-M_t)-x(t-M_t){|}_{P_1}^2\\&+4\lambda_{\text{max}}(P_2,P_1)\Big( \int_{t-M_t}^{t}e^{-\eta (t-\tau)} {|}w(\tau){|}_{Q_w}^2 d \tau\\& + \sum_{\tau \in [t-M_t,t]\cap K_s}e^{-\eta (t-\tau)} {|}v(\tau){|}_{Q_v}^2\Big).
		\end{aligned}
	\end{align*}
	Selecting the horizon length  $M$ large enough such that 
	\begin{align*}
		e^{-\tilde{\eta}M}:=4\lambda^2_{\text{max}}(P_2,P_1)e^{-\eta{M}}<1
	\end{align*}
	with $\tilde\eta \in \mathbb{R}_{\geq0}$, we obtain for all $t\geq M$
	\begin{align}
		\begin{aligned}
			{|}\hat{e}(t){|}_{P_1}^2&\leq e^{-\tilde{\eta}M_t} {|}\hat{e}(t-M_t){|}_{P_1}^2\\&+4\lambda_{\text{max}}(P_2,P_1)\Big( \int_{t-M_t}^{t}e^{-\eta (t-\tau)} {|}w(\tau){|}_{Q_w}^2 d \tau \\&+ \sum_{\tau \in [t-M_t,t]\cap K_s}e^{-\eta (t-\tau)} {|}v(\tau){|}_{Q_v}^2\Big).
		\end{aligned}
		\label{eq:boundMstep}
	\end{align}
	Consider an arbitrary $t\geq M$ and note that it can be decomposed as $t=l+T$ with $l\in [0,M)$ and $T\geq M$ as specified below. 
	Using (\ref{eq:boundMtstep1}) we obtain
	\begin{align}
		\hspace{-0.6em}	
		\begin{aligned}
			{|}\hat{e}(l){|}_{P_1}^2 \leq & 4 \lambda_{\text{max}}(P_2,P_1)\big( e^{- \eta l} {|}\hat{e}(0){|}_{P_2}^2\\&+\int_{0}^{l}e^{-\eta(l-\tau)} {|}w(\tau){|}_{Q_w}^2 d\tau\\&+\sum_{\tau\in [0,l]\cap K_s}e^{-\eta(l-\tau)} {|}v(\tau){|}_{Q_v}^2\big).
		\end{aligned}
		\label{eq:boundlk}
	\end{align}
	The  $T$-long time interval consists of $\kappa$ time intervals $[k_{i+1},k_i]$, $i=1,\ldots, \kappa$ with \mbox{$k_{i+1}=k_i-M-\delta_{k_i}$} and $k_{1}=t$. Recall $\delta_{k_i}=k_i-\max\{0,j \in K_s|j\leq k_i\}$.
	Applying (\ref{eq:boundMstep})  for each of the $\kappa$ time intervals and using $\eta \leq \tilde\eta $ yields
	\begin{align}
			{|}\hat{e}(k_i){|}_{P_1}^2\leq& e^{-\tilde\eta(M+\delta_{k_i})} {|}\hat{e}(k_{i+1}){|}_{P_1}^2\nonumber\\&+4\lambda_{\text{max}}(P_2,P_1)\Big(\int_{k_i-M-\delta_{k_i}}^{k_i}e^{-\tilde\eta(k_i-\tau)} {|}w(\tau){|}_{Q_w}^2 d\tau \nonumber\\&+\sum_{\tau\in[k_i-M-\delta_{k_i},k_i]\cap K_s}e^{-\tilde\eta(k_i-\tau)} {|}v(\tau){|}_{Q_v}^2\Big). 	\label{eq:boundMstepk}
	\end{align}
	Applying (\ref{eq:boundMstepk}) recursively we derive the following
	upper bound for the estimation error
	\begin{align}
			&{|}\hat{e}(t){|}_{P_1}^2 \nonumber\\\leq\phantom{.}& e^{-\tilde\eta{T}}	{|}\hat{e}(l){|}_{P_1}^2+4\lambda_{\text{max}}(P_2,P_1)\Big(\int_{l}^{t}e^{-\tilde\eta(t-\tau)}{|}w(\tau){|}_{Q_w}^2 d\tau\nonumber\\ &+\sum_{\tau\in[l,t]\cap K_s} e^{-\tilde\eta(t-\tau)}2{|}v(\tau){|}_{Q_v}^2\Big) \nonumber\\
			\stackrel{(\ref{eq:boundlk})}{\leq}& 4\lambda_{\text{max}}(P_2,P_1)e^{-\tilde\eta{T}}\Big(e^{-\eta l} {|}\hat{e}(0){|}_{P_2}^2\nonumber\\&+\int_{0}^{l}e^{-\eta(l-\tau)} {|}w(\tau){|}_{Q_w}^2 d\tau	\label{eq:boundkMstepk}
			\\&+\sum_{\tau\in [0,l]\cap K_s}e^{-\eta(l-\tau)} {|}v(\tau){|}_{Q_v}^2\Big) \nonumber\\&+4\lambda_{\text{max}}(P_2,P_1)\Big(\int_{l}^{t}e^{-\tilde\eta(t-\tau)}{|}w(\tau){|}_{Q_w}^2 d\tau\nonumber\\&+\sum_{\tau\in[l,t]\cap K_s} e^{-\tilde\eta(t-\tau)}2{|}v(\tau){|}_{Q_v}^2\Big).\nonumber
	\end{align}
	Note that the factor 2 in the summation term over the interval $[l,T]$ arises because recursively applying (\ref{eq:boundMstepk}) causes $v(\tau)$ to appear twice in the final bound for every  $\tau \in \{k_i\}_{i=2}^{\kappa-1}$.
	Due to $\eta\leq \tilde\eta$ and $t=l+T$ we can write
	\begin{align}
		\begin{aligned}
			{|}\hat{e}(t){|}_{P_1}^2&\leq 4 \lambda_{\mathrm{max}}(P_2,P_1)\Big(e^{-\tilde\eta t}{|}\hat{e}(0){|}_{P_2}^2
			\\&+\int_{0}^{t} e^{-\tilde\eta (t-\tau)}  {|}w(\tau){|}_{Q_w}^2 \ d\tau\\&+\sum_{\tau \in [0,t] \cap K_s}
			e^{-\tilde\eta (t-\tau)}  2{|}v(\tau){|}_{Q_v}^2\Big)
		\end{aligned}
		\label{eq:boundsum}
	\end{align}
	Note that since in (\ref{eq:boundkMstepk}) $e^{-\tilde\eta(t-l)}{|}v(l){|}_{Q_v}^2$ is already considered with the factor 2 in the interval $[l,T]$, we can omit \emergencystretch=2em $e^{-\tilde\eta T}{|}v(l){|}_{Q_v}^2$ from (\ref{eq:boundkMstepk}) to obtain (\ref{eq:boundsum}).
	For the case $t<M$, (\ref{eq:boundsum}) can be obtained directly from (\ref{eq:boundMtstep1}). We conclude that (\ref{eq:boundsum}) holds for all $t \geq 0$, completing the proof.	\qed                                        
\end{document}